\documentclass{article}

\usepackage{microtype}

\usepackage[tableposition=top]{caption}
\usepackage{subcaption}
\usepackage{graphicx}
\usepackage{macros}
\usepackage{booktabs} 
\usepackage{multirow}
\usepackage[dvipsnames]{xcolor}  
\definecolor{mreg}{RGB}{25,23,124}
\newcommand{\regcolor}[1]{\textcolor{mreg}{#1}}
\definecolor{mop}{RGB}{7, 0, 255}
\newcommand{\opcolor}[1]{\textcolor{mop}{#1}}
\definecolor{mconst}{RGB}{113,113,113}

\usepackage{listings}
\usepackage{newfloat}
\usepackage{multirow}
\DeclareFloatingEnvironment{listin}
\SetupFloatingEnvironment{listin}{name=Listing}
\newenvironment{code}{\captionsetup{type=listin}}{}
\addtocounter{listin}{-1}
\usepackage{amsmath}
\usepackage{amsthm}
\usepackage{thm-restate}

\declaretheorem[name=Theorem]{theorem}
\usepackage[textsize=tiny]{todonotes}
\usepackage{algorithm}
\usepackage{algorithmic}
\usepackage{amssymb}
\usepackage{enumitem}
\usepackage{hyperref}
\usepackage[numbers]{natbib}

\usepackage[accepted]{mlsys2024}

\mlsystitlerunning{COMET: Neural Cost Model Explanation Framework}

\begin{document}
\twocolumn[
\mlsystitle{COMET: Neural Cost Model Explanation Framework}



\mlsyssetsymbol{equal}{*}

\begin{mlsysauthorlist}
\mlsysauthor{Isha Chaudhary}{UIUC}
\mlsysauthor{Alex Renda}{MIT}
\mlsysauthor{Charith Mendis}{UIUC}
\mlsysauthor{Gagandeep Singh}{UIUC,vmware}
\end{mlsysauthorlist}

\mlsysaffiliation{UIUC}{Department of Computer Science, University of Illinois Urbana-Champaign, Illinois, USA}
\mlsysaffiliation{MIT}{MIT CSAIL, Massachusetts, USA}
\mlsysaffiliation{vmware}{VMWare Research, USA}

\mlsyscorrespondingauthor{Isha Chaudhary}{isha4@illinois.edu}

\mlsyskeywords{Machine Learning, MLSys}

\vskip 0.3in
]

\begin{abstract}
Cost models predict the cost of executing given assembly code basic blocks on a specific microarchitecture. Recently, neural cost models have been shown to be fairly accurate and easy to construct. They can replace heavily engineered analytical cost models used in mainstream compiler workflows. However, their black-box nature discourages their adoption. In this work, we develop the first framework, \tool{}, for generating faithful, generalizable, and intuitive explanations for neural cost models. We generate and compare \tool{}'s explanations for the popular neural cost model, Ithemal against those for an accurate CPU simulation-based cost model, uiCA. Our empirical findings show an inverse correlation between the prediction errors of Ithemal and uiCA and the granularity of basic block features in \tool{}'s explanations for them, thus indicating potential reasons for the higher error of Ithemal with respect to uiCA. 
\end{abstract}
\printAffiliationsAndNotice{}  
\section{Introduction}
\emph{Cost models} predict the cost (memory, time, energy, etc) that an assembly code basic block, a sequence of assembly instructions with no jumps or loops, takes while executing on a specific microarchitecture. They are used to guide compiler optimization~\citep{vectorization, opt2} and superoptimization~\citep{stoke}. They can be traditional or learned models. 

Traditional cost models are typically simulation and static-analysis based models. Simulation-based cost models generate their predictions by simulating program execution for a given CPU. 
They are hand-engineered using released documentation and micro-benchmarking the CPU under study.
Examples of these cost models are uiCA~\citep{uica} and LLVM-MCA~\citep{llvm-mca}. 
As these models are traditional programs, domain experts can intuitively understand and debug them, and hence they are commonly deployed in practical settings. However, they require significant engineering effort to construct and must be manually re-engineered to reflect changes across CPU microarchitectures.
Static-analysis based cost models, such as IACA~\citep{iaca} and OSACA~\cite{osaca} predict the cost of executing a program on a given microarchitecture by developing a model of the target CPU and deploying static analysis methods to predict the cost of a given basic block with the model. These cost models often have a high error in their predictions~\cite{uica, bhive}. 
%

Alternatively, machine learning techniques can be used to learn a cost model~\citep{ithemal, tpu-cost-model, dlccostmodel,granite}.
Development of ML-based cost models requires the one-time effort of collecting a dataset of representative programs, collecting the end-to-end cost for the execution of those programs on the CPU under study, and training a selected type of ML model. While simple and interpretable ML models could be used for constructing cost models, prior work~\cite{ithemal, granite, tpu-cost-model, dlccostmodel} has used neural networks as cost predictors to precisely approximate the complex function mapping basic blocks to their costs. 
An instance of such neural cost models is Ithemal~\citep{ithemal}, which is an LSTM model trained on the BHive~\citep{bhive} dataset of x86 basic blocks to predict basic block throughput (average number of CPU clock cycles to execute the block when looped in steady state).
Ithemal is more accurate on the BHive dataset than most throughput models~\citep{bhive}. Ithemal needs less manual effort to construct than any simulation-based or static-analysis-based cost model. 
However neural models generally have the downside that they are uninterpretable~\cite{interpretmlbook}. 

\textbf{This work.}
Our goal is to bring interpretability to inherently black-box but accurate neural cost models, by developing a general framework that can generate trustworthy and intuitive explanations of their predictions. These neural cost models could have arbitrary architectures~\cite{ithemal,granite}, requiring custom explanation methods, and could also be proprietary. To avoid engineering custom explanation methods for each model, we develop a common explanation framework that is agnostic to the type or structure of the model. Apart from saving manual engineering effort, a common framework would facilitate a comparison between neural and other types of cost models with respect to the explanations of their predictions. To achieve our goals, we develop our explanation framework to generate explanations that (i) assume just query-access to the cost model, which may be available for some proprietary cost models too, (ii) faithfully  reflect the cost model's behavior, 
(iii) generalize across multiple basic blocks, and 
(iv) are simple and interpretable for domain experts. 

\textbf{Key challenges.} 
%
For building trustworthy explanations, we need to formalize
the desirable properties of faithfulness, generalizability, and simplicity~\citep{explanationprops}.  
%
There is a tradeoff between the degree to which a given explanation satisfies the above desirable properties and its computational cost. Therefore, we need to design efficient algorithms that can balance this tradeoff.
Prior works~\citep{lime,anchors} in domains such as Vision or NLP have used perturbed inputs to efficiently generate explanations with only query-access to the model. However, their perturbation algorithms heavily utilize local neighborhoods of their inputs while creating their explanations. In the discrete domain of basic blocks, there is no well-defined concept of locality. Hence, we need specialized perturbation algorithms to handle this domain-specific challenge and derive close approximations to the complex behavior of a given cost model in a reasonable number of queries. 

\textbf{Our approach.}
We identify that global explanations with desirable properties may be computationally intractable or even may not exist for complex cost models. Hence, we focus on explaining a given model's prediction for a target basic block. We first formalize the ideal, query-based, block-specific explanations with desirable properties as an optimization problem. We observe that generating such ideal explanations is intractable. To practically generate explanations, we relax our requirements and develop \tool{}, a perturbation-based explanation framework based on the design of (i) novel primitives for explanations that capture both coarse-grained (e.g. number of instructions) and fine-grained (e.g., instructions and data dependencies) features of the basic block, and (ii) new custom perturbation algorithms for generating a diverse set of basic blocks that help gauge the complex behaviors of cost models.

\textbf{Contributions.}
We make the following contributions:
\begin{enumerate}[leftmargin=*]
    \item We formalize the ideal query-based explanations having desirable properties for cost models for any target basic block as an 
    optimization problem that is agnostic to any particular Instruction Set Architecture (ISA). 
    \item  We relax the problem to make it practically solvable. Building on our relaxation, we present \tool{} (COst Model ExplanaTion framework), a novel and efficient explanation framework for neural cost models. As \tool{} depends on the ISA, we have implemented it for the popular x86 ISA, and it can be extended to other ISAs with non-trivial engineering effort. We open-source our implementation at \url{https://github.com/uiuc-focal-lab/COMET}. \tool{}'s explanations identify the features of a target basic block that are important for a given cost model's prediction. 
    \item We systematically analyze \tool{}'s accuracy and use it to gain insights into the working of common cost models. We explain basic blocks in the popular BHive dataset~\citep{bhive}. We empirically observe that \tool{}'s explanations for the neural cost model Ithemal more often consist of coarser-grained features of the basic block, such as the block's number of instructions, as compared to the explanations for the lowest error simulation-based cost model uiCA, indicating potential sources of the relatively higher error in Ithemal's predictions with respect to uiCA. 
\end{enumerate}
%
\tool{} aims to help our stakeholders, i.e. compiler and performance engineers, develop an intuition about and debug neural cost models in a simple yet precise way. 
We anticipate this work to go a long way in developing better neural cost models and making them trustworthy.


\section{Related work}


\textbf{Explanation techniques}.
Explanations for ML models consist of either building inherently interpretable ML models~\citep{inherentinterpret1} or creating post-hoc explanations for the models~\citep{lime, anchors, posthoc, expalgo}. Post-hoc explanations are preferred as accurate cost modeling for the CPU's pipelined architecture makes complex models more suitable. These can either describe a model globally~\citep{shap} or for specific inputs~\citep{lime, anchors}. Explanation techniques can also be broadly classified as black-box~\citep{lime, anchors, shap} and white-box techniques~\citep{saliency1, saliency2}. Further classifications of explanation techniques can be as perturbation/example-based ~\citep{pertxai3,pertxai2, pertxai1} and symbolic explanation techniques~\citep{symb1,symb2,symb3,symb4}. While symbolic methods give formal guarantees on the explanations, they do not scale to complex models yet. 
\citep{anica} is a differential-testing tool to analyze the inconsistencies between multiple cost models. This tool, unlike \tool{}, is not meant to explain a particular prediction of a cost model to enable case analysis.

\textbf{Input perturbation algorithms}.
For domains wherein the input is a sequence of discrete entities such as NLP and code, the prior perturbation-based explanation algorithm by \citet{anchors} has used generative models \citep{bert, codebert} to obtain input perturbations. These perturbations might not be syntactically correct and can result in erroneous explanations~\citep{explanationcode}. Hence, we have not used such unconstrained perturbation techniques in our explanation framework. Moreover, as mentioned above, there is no well-defined concept of locality in this domain. Thus, we can not use the perturbation algorithms from prior work in other domains which generally perturb the input in some local regions.
Stoke \citep{stoke} is a stochastic superoptimizer that perturbs input x86 assembly programs to optimize them. While Stoke does not operate on embedding spaces, it can generate syntactically incorrect perturbations. 

\section{Overview and motivating example}
\begin{figure*}[htb]
    \centering
    \includegraphics[width=\textwidth]{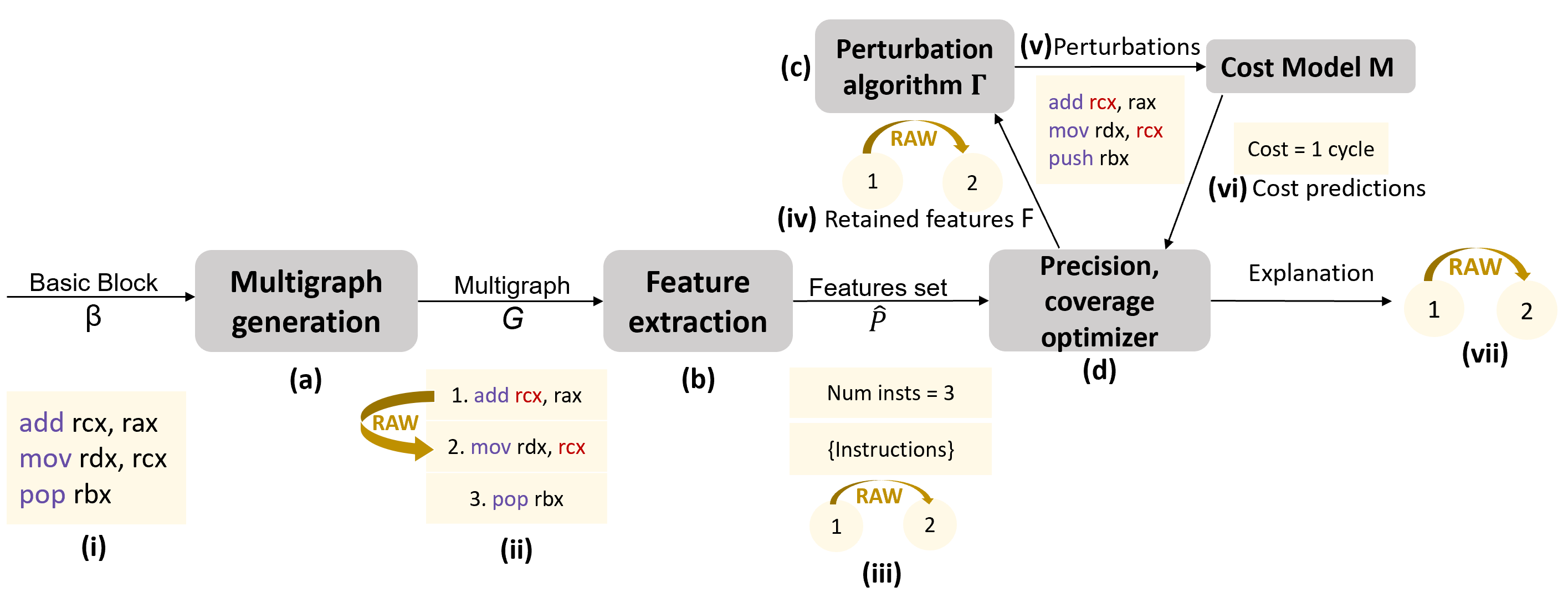}
    \caption{ \tool{} is given a cost model $\mathcal{M}$ and a basic block $\basicblock$ as input. \tool{} identifies the features of $\basicblock$ that explain the prediction $\costmodel(\basicblock)$. \tool{} first converts $\basicblock$ to a multigraph $\mathcal{G}$ in (a). $\mathcal{G}$ has the instructions and data dependencies of $\basicblock$ as its vertices and edges respectively. The features $\Hat{\mathcal{P}}$ of $\basicblock$,  are then extracted from $\mathcal{G}$ in (b). Sets of these features are fed into \tool{}'s perturbation algorithm $\Gamma$ (c) that generates several perturbations that preserve the features in the corresponding retained input feature set $\mathcal{F}$. \tool{} obtains the predictions of cost model $\costmodel$ for each perturbed basic block, which are then used for the estimation of the precision and coverage of a feature set $\mathcal{F}$. $\mathcal{F}$ having precision higher than $(1-\delta)$ with maximum coverage is identified by the precision and coverage optimizer in (d) and is output as \tool{}'s explanation for $\costmodel(\basicblock)$.}
    \label{fig:overview}
\end{figure*}

Consider the x86 basic block in Listing~\ref{code:running_eg}(a), represented in the Intel assembly syntax~\citep{intelsyntax} (default syntax throughout the paper). The throughput cost of this basic block predicted by the neural cost model Ithemal for the Haswell microarchitecture is $1.3$ cycles. We use \tool{} to explain the prediction made by Ithemal. 
An \emph{explanation} is a small set of features of the input basic block whose presence is \emph{sufficient} to get a prediction from the cost model close to its original prediction. 
\tool{} identifies the feature corresponding to Read-after-Write (RAW) data dependency between instructions 1 and 2 as an explanation for Ithemal's throughput prediction. Please refer to Appendix~\ref{app:deps} for information on the types of common data dependencies in assembly basic blocks. 

\begin{code}
\begin{center}
\begin{minipage}{0.2\textwidth}
\centering
\begin{lstlisting}[language={[x86masm]Assembler}]
add rcx, rax
mov rdx, rcx
pop rbx
\end{lstlisting}
\subcaption{Input basic block}
\end{minipage}
\begin{minipage}{0.2\textwidth}
\centering
\begin{lstlisting}[language={[x86masm]Assembler}]
add rcx, rax
mov rdx, rcx
|\textbf{\color{red}push}| rbx
\end{lstlisting}
\subcaption{Perturbed basic block}
\end{minipage}
\captionof{listin}{Motivating Example}
\label{code:running_eg}
\end{center}
\end{code}

Figure~\ref{fig:overview} illustrates \tool{}'s working on the above example. 
\tool{} first extracts the candidate features, i.e., individual instructions, their count, and data dependencies between instruction pairs, as shown in Figure~\ref{fig:overview}(iii). 
It searches over all possible combination sets of the candidate features for an explanation that is simple, faithful to the cost model's behavior, and extends to other basic blocks as well. 
\tool{} evaluates any feature set through perturbations of the input basic block that retain the features in the set [Figure~\ref{fig:overview}(c)]. 
For example, for the singleton set consisting of just the RAW dependency in the block, a possible perturbation is shown in Listing~\ref{code:running_eg}(b), which preserves the dependency. 
\tool{} obtains the predictions of the given cost model for each perturbed basic block and uses them for the estimation of the faithfulness (precision) and generalizability (coverage) of a set of features. The set with precision higher than a threshold and maximum coverage is output as \tool{}'s explanation for the given cost model and basic block [Figure~\ref{fig:overview}(d)]. 

As per~\citet{uops.info}, for Haswell, the canonical forms of the individual instructions in the block have similar throughputs. The RAW data dependency will disable the independent, parallel execution of instructions 1 and 2, and hence intuitively be a bottleneck for the block's execution. 
As \tool{}'s explanation matches our intuition, we can be confident that Ithemal predicts throughput focusing on the correct features for this and similar blocks. 

\section{Formalizing cost model explanations}   
In this section, we formalize our notion of explanations for cost models and discuss the desirable properties of the explanations. 
Our objective is to develop an explanation framework that only queries a given cost model (query-access only) to explain its behavior. Table~\ref{tab:notation} presents the important notation used when discussing the following formalism. 

\begin{table}
    \centering
    \caption{Notation}
    \begin{tabular}{@{}lr@{}}
        \toprule
          Notation & Meaning\\
         \midrule
         $\costmodel$ & Cost model\\
         $\basicblock$ & Basic block\\
         $\mathcal{P}$ & Set of all block features\\
         $\hat{\mathcal{P}}$ & Block features used to form explanations\\
         $\mathcal{F}$ & Given set of features\\
         $Prec(\mathcal{F})$ & Precision of $\mathcal{F}$\\
         1-$\delta$ & Precision threshold\\
         $Cov(\mathcal{F})$ & Coverage of $\mathcal{F}$\\
         $\pertmodel(\mathcal{F})$ & All perturbations of $\basicblock$ retaining $\mathcal{F}$\\
         $\mathcal{D}_\mathcal{F}$ & Distribution over blocks in $\pertmodel(\mathcal{F})$\\
         $\graphbb$ & Multigraph of $\basicblock$\\
         $\Gamma$ & \tool{}'s perturbation algorithm\\
         \bottomrule
    \end{tabular}
    
    \label{tab:notation}
    
\end{table} 



We first formalize the cost model as a function $\costmodel$ that maps valid basic blocks in a given Instruction Set Architecture (ISA) to real-valued costs.  
%
Let $\mathcal{T}$ be a set of $\costmodel$'s predictions that we want to explain over the set of valid basic blocks (global explanation). 

An explanation for the behavior of $\costmodel$ over $\mathcal{T}$ is the common features of basic blocks having cost prediction in $\mathcal{T}$, that are not present in other basic blocks.

For example, consider a hypothetical, crude throughput-predicting cost model $\costmodel_1$ that assigns a throughput of $2$ cycles if and only if a basic block has $8$ instructions. If $\mathcal{T} = \{2\}$, then all the blocks having prediction in $\mathcal{T}$ will have $8$ instructions while the rest will not. Hence the distinguishing factor and therefore a correct global explanation of the cost model's behavior in $\mathcal{T}$ will be the number of instructions equal to $8$. 

While it might be possible for simple cost models such as $\costmodel_1$ to have global concepts for their predictions that can explain them in certain $\mathcal{T}$, generally cost models are complex and specialized to exploit features of the input basic block. Hence, as it may be infeasible to generate global explanations with desirable properties, we focus on generating block-specific explanations for cost model $\costmodel$. 
To explain the cost prediction for the basic block $\basicblock$, we specialize $\mathcal{T}$ to be an $\epsilon-$ball around $\costmodel(\basicblock)$,
where $\epsilon> 0$ is a small constant. 
Our generated explanations are sets of features of the block whose presence is sufficient for the cost model to make its prediction.


The smallest meaningful units (basic features) of an assembly basic block are its tokens (opcodes and operands). 
Let $\mathcal{P}_\basicblock$ be the set of all basic features and all functions of basic features, which we cumulatively call features, of the basic block $\basicblock$. Some elements of $\mathcal{P}_\basicblock$ for the input basic block in Figure~\ref{fig:overview} are shown in Figure~\ref{fig:overview}(iii).
Note that, as $\mathcal{P}_\basicblock$ captures all the features of the basic block, there is a one-to-one correspondence between $\mathcal{P}_\basicblock$ and the block. 
The remaining discussion in this paper will describe our approach for generating explanations for the cost model $\costmodel$ when predicting the cost for a basic block $\basicblock$. 
To simplify notation, unless mentioned otherwise, we will omit $\costmodel$ and $\basicblock$ from the subscripts of symbols, e.g., $\mathcal{P}_\basicblock$ will be written as $\mathcal{P}$. 



\subsection{Ideal query-only explanation}\label{sec:idealexp} 

Let the set of features $\mathcal{F}^*\subseteq \mathcal{P}$ be the ideal explanation for $\costmodel(\basicblock)$ based on queries to $\costmodel$.
Desirable properties of $\mathcal{F}^*$ are that it should be \emph{faithful} to cost model's behavior, \emph{generalizable} to multiple basic blocks, and \emph{simple}~\citep{explanationprops}, which we formalize next. 


Let $\pertmodel$ be a perturbation function that is given a set of features $\mathcal{F}$ of basic block $\basicblock$ as input, and returns a set of valid assembly basic blocks $B_{\mathcal{F}}$, where each basic block $\basicblock'\in B_{\mathcal{F}}$ differs from $\basicblock$ only by some perturbations to the features in $\mathcal{P}\setminus\mathcal{F}$ in $\basicblock$. 
Output of $\pertmodel(\mathcal{F})$ includes $\basicblock$ also.
Consider the block in Figure~\ref{fig:overview}(i). If the set \{instruction 1: \opcolor{add} \regcolor{rcx} \regcolor{rax}\} is input into $\pertmodel$, then the basic block shown in Figure~\ref{fig:overview}(v) is an element in the output set of basic blocks, as it perturbs some features not in the input set of features. 


\textbf{Faithfulness}.
A set of features $\mathcal{F}\subseteq \mathcal{P}$ will be a faithful explanation for the prediction of $\costmodel(\basicblock)$ if perturbations of features of $\basicblock$ that are not in $\mathcal{F}$ cannot change the cost prediction of $\costmodel$ significantly, i.e., by more than $\epsilon$. Otherwise, $\mathcal{F}$ does not completely capture the features used by $\costmodel$ for its prediction for $\basicblock$, and hence is not faithful to $\costmodel$'s behavior. \eqref{eq:faithful} presents the above notion as a logical statement $\varphi(\mathcal{F})$ that must be satisfied by the ideal, faithful explanation $\mathcal{F}^*$. 
\begin{equation}
    \label{eq:faithful}
    \varphi(\mathcal{F}) \triangleq 
    (\mathcal{F}\subseteq\mathcal{P}) \text{ and }
    (\forall\alpha\in \pertmodel(\mathcal{F})\ldotp \costmodel(\alpha) \in\mathcal{T})
\end{equation}


A trivially faithful set of features is $\mathcal{P}$, as it would contain all the basic block features that are important for cost prediction. 
But this explanation is not useful, as $P$ can faithfully explain $\basicblock$ for any cost model but it does not precisely distinguish features according to the behavior of a target cost model.   


\textbf{Generalizability and simplicity}.
To overcome the above issue, we require that faithful explanations of basic block $\basicblock$ should also explain other blocks that contain the features in the explanation and where the cost model $\costmodel$ makes predictions close to $\costmodel(\basicblock)$ (generalizable).
Every set $\mathcal{F}\subseteq\mathcal{P}$ will have a corresponding set of basic blocks (potentially empty), $\Omega_{\mathcal{F}}$~\eqref{eq:maxgenset} containing basic blocks with similar predictions as $\basicblock$ and having $\mathcal{F}$ as faithful explanations. 
For faithful explanations with maximum generalizability, we need to maximize the cardinality of $\Omega_{\mathcal{F}}$ over the set of faithful explanations. 
\begin{equation}
    \label{eq:maxgenset}
    \Omega_{\mathcal{F}} \triangleq \{\alpha\in\pertmodel(\mathcal{F})\text{ and }\costmodel(\alpha)\in\mathcal{T} \text{ and }\varphi_\alpha(\mathcal{F})\}
\end{equation}

For higher interpretability, ideal explanation $\mathcal{F}^*$ should be simple. While there are many metrics for simplicity of explanations, a common metric for sets of features used as explanations is their cardinality~\citep{explanationsize,interpretmlbook}. 
Hence, for simple, faithful, and generalizable explanations $\mathcal{F}^{*}$, we solve the optimization problem~\eqref{eq:overallideal}, where $\lambda > 0$ is a regularization parameter. 

\begin{equation}
    \label{eq:overallideal}
    \mathcal{F}^* \triangleq \mathop{argmax}_{\mathcal{F} \text{ s.t. } \varphi(\mathcal{F})} (\lvert \Omega_{\mathcal{F}}\rvert - \lambda.\lvert \mathcal{F}\rvert)
\end{equation}

\subsection{Practical query-only explanations}\label{sec:practicalexp} 


There are two levels of intractability in the above formulation of ideal explanations~\eqref{eq:overallideal}. First, the evaluation of the faithfulness condition~\eqref{eq:faithful} for a given set of features $\mathcal{F}$ requires querying $\costmodel$ for the cost prediction of all the basic blocks in the large set, $\pertmodel(\mathcal{F})$. We refer the reader to Appendix~\ref{app:samplespacesize} for examples of some estimates of the cardinality of $\pertmodel(\mathcal{F})$. Second, the computation in~\eqref{eq:maxgenset} requires predicting the cost and computing faithful explanations for all basic blocks in $\pertmodel(\mathcal{F})$. Hence, to practically solve the explanation problem, we relax it as described next. 

\textbf{Probabilistic faithfulness}. 
To simplify the faithfulness condition in~\eqref{eq:faithful}, we relax the requirement of the cost prediction for all perturbed basic blocks to be in $\mathcal{T}$ with the requirement that the probability of the cost of perturbed blocks being in $\mathcal{T}$ to be higher than a threshold. This threshold will denote the degree of faithfulness of our explanations. This probability can be represented as $Pr_{\alpha\sim\mathcal{D}_{\mathcal{F}}}(\costmodel(\alpha)\in\mathcal{T})$, where $\mathcal{D}_{\mathcal{F}}$ is a distribution over all perturbed basic blocks that retain the features in $\mathcal{F}$, $\pertmodel(\mathcal{F})$.  
We identify that the probability is analogous to \emph{precision}~\eqref{eq:precision} used in prior work~\citep{anchors}, and hence we adopt this terminology.  
Thus, probabilistic faithful explanations $\mathcal{F}$ must satisfy the condition $\hat{\varphi}(\mathcal{F})$, given by~\eqref{eq:probfaithful}, where $0\leq\delta\leq1$. 

\begin{equation}
    \label{eq:precision}
    Prec(\mathcal{F}) \triangleq Pr_{\alpha\sim\mathcal{D}_{\mathcal{F}}}(\costmodel(\alpha)\in\mathcal{T})
\end{equation}
\begin{equation}
    \label{eq:probfaithful}
    \hat{\varphi}(\mathcal{F}) \triangleq (\mathcal{F}\subseteq\mathcal{P})\text{ and }(Prec(\mathcal{F})\geq(1-\delta))
\end{equation}
As the distribution over basic blocks, $\mathcal{D}_\mathcal{F}$ in~\eqref{eq:precision} should be such that $\hat{\varphi}(\mathcal{F})$ closely approximates the ideal faithfulness condition~\eqref{eq:faithful} which has no prioritization over the perturbed basic blocks, it should ideally be a uniform distribution over its sample space $\pertmodel(\mathcal{F})$ and hence depend on $\mathcal{F}$. 

\textbf{Probabilistic generalizability and simplicity}. 
To relax the computation in~\eqref{eq:maxgenset} we overapproximate it with the perturbed basic blocks' set, $\pertmodel(\mathcal{F})$. Thus, for higher generalizability, we maximize $\lvert\pertmodel(\mathcal{F})\rvert$. 
Note that $\pertmodel$ is a monotonically decreasing function (proof in Appendix~\ref{app:proionofs}). 
Thus, for simplicity of explanations too we can 
maximize $\lvert\pertmodel(\mathcal{F})\rvert$. 
We normalize $\lvert\pertmodel(\mathcal{F})\rvert$ with the number of all possible perturbations of the basic block, $\lvert\pertmodel(\emptyset)\rvert$ where $\emptyset$ denotes an empty set of features to preserve. $\pertmodel(\emptyset)$ is independent of $\mathcal{F}$ and hence the normalization will not affect the optimization problem's output. Intuitively, the resultant fraction in the optimization objective would denote the fraction of all possible perturbations that preserve the feature set $\mathcal{F}$. We relax this computation by replacing it with the probability of finding the features in $\mathcal{F}$ in a randomly selected valid perturbation of the basic block. We identify that this probability is analogous to \emph{coverage} in prior work~\citep{anchors}, and hence we adopt this terminology. 
Coverage constitutes a probabilistic notion of generalizability and simplicity of explanations, and hence we maximize the coverage in our optimization objective. \eqref{eq:coverage} defines the coverage of a set of features $\mathcal{F}$, where $\mathcal{D}$ is a distribution over all perturbations of the input basic block, $\pertmodel(\emptyset)$. To obtain an unbiased measure of coverage, $\mathcal{D}$ should ideally be a uniform distribution. 
\begin{equation}
    \label{eq:coverage}
    Cov(\mathcal{F}) \triangleq Pr_{\alpha\sim\mathcal{D}}(\mathcal{F}\subseteq\mathcal{P}_\alpha)
\end{equation}


\textbf{Overall practical optimization problem}. 
Thus, our optimization problem to practically find the desirable explanation $\Hat{\mathcal{F}}^{*}$ for  $\costmodel(\basicblock)$ becomes~\eqref{eq:overallapproxpractical}.
\begin{equation}
    \label{eq:overallapproxpractical}
    \Hat{\mathcal{F}}^{*} \triangleq \mathop{argmax}_{\mathcal{F}\text{ s.t. } \mathcal{F}\subseteq\mathcal{P}}  \lvert Cov(\mathcal{F})\rvert \text{ s.t. } Prec(\mathcal{F})\geq(1-\delta)
\end{equation}



\section{COMET: Neural Cost Model Explanation Framework}\label{sec:implementation}
This section presents \tool{}, our novel framework for efficiently generating desirable explanations for the predictions made by a given cost model for a target basic block. 
%
The core operation of \tool{} is to efficiently solve the optimization problem in~\eqref{eq:overallapproxpractical}. While COMET is not conceptually limited to any particular Instruction Set Architecture (ISA), we develop the framework for one of the most popular ISA \textemdash x86 and leave the development for other ISAs to future work. 
An overview of \tool{}'s algorithm on an example x86 basic block is shown in Figure~\ref{fig:overview}. To create interpretable explanations, \tool{} first decomposes the input basic block $\basicblock$ into its features. 
We restrict $\mathcal{P}$, which consists of all possible features of a basic block, to block features $\Hat{\mathcal{P}}\subset\mathcal{P}$ [Section~\ref{sec:explanationfeatures}], to reduce the possible sets of features to evaluate in the optimization problem in~\eqref{eq:overallapproxpractical}. 
%
Next, we need to evaluate the precision of each subset of $\Hat{\mathcal{P}}$ to identify a subset $\mathcal{F}$ that has $Prec(\mathcal{F})\geq(1-\delta)$ and has maximum $Cov(\mathcal{F})$. To efficiently solve this constrained optimization problem, \tool{} adapts the Anchors explanation algorithm~\citep{anchors}, which has a similar optimization objective [Section~\ref{sec:anchoradapt}]. 

\subsection{Extracting block features}\label{sec:explanationfeatures}

\tool{} casts the input basic block into a multigraph $\graphbb = (\mathcal{V},\mathcal{E})$, which we describe next. Figure~\ref{fig:overview}(ii) shows the multigraph for the example block. 
We define $\mathcal{V}$ as the set of vertices of the multigraph, corresponding to all the instructions, annotated with their positions in the block. 
%
$\mathcal{E}$ consists of directed edges between instructions that have data dependencies, labeled by the types of data dependency hazards between them. Please refer to Appendix~\ref{app:deps}  for information on the different types of data dependencies commonly found in assembly basic blocks. 
Figure~\ref{fig:overview}(ii) shows the RAW type of hazard that is present in the example block.  Separate edges for different dependencies help identify specific dependencies as bottlenecks in explanations.
Distinguishing dependencies can be crucial to debugging cost models. If a cost model is found to base its predictions on dependencies that are optimized by the compiler, developers can debug the cost model to eradicate any spurious correlations.

We constitute $\Hat{\mathcal{P}}$ with the instructions, data dependencies, and number of instructions of the block. 
Figure~\ref{fig:overview}(iii) shows the features in $\Hat{\mathcal{P}}$ for the example basic block. These features are used in the design of popular hand-engineered cost models~\citep{uica,iaca,llvm-mca}, correspond better with the neural cost models~\citep{ithemal}, and hence are interpretable for our stakeholders. We restrict to these common features used in popular cost models to make \tool{} focus on the important set of features and generate explanations efficiently. 



\subsection{Efficiently computing explanations}\label{sec:anchoradapt}

To efficiently compute explanations, \tool{} empirically estimates $Prec(\mathcal{F})$ and $Cov(\mathcal{F})$ with samples from basic block distributions, $\mathcal{D}_{\mathcal{F}}$ and $\mathcal{D}$ respectively. 
%
%
%
%
We have designed \emph{basic block perturbation algorithms} to sample from $\mathcal{D}_{\mathcal{F}}$ and $\mathcal{D}$, which essentially perturb basic block $\basicblock$ to obtain blocks $\basicblock'$ according to the corresponding distribution from the underlying sample space. As discussed in Section~\ref{sec:practicalexp}, we want both $\mathcal{D}_{\mathcal{F}}$ and $\mathcal{D}$ to be uniform distributions over their respective sample spaces to compute unbiased approximations of the ideal desirable explanations. Observe that, $\mathcal{D}$ is hence a special case of $\mathcal{D}_{\mathcal{F}}$ with $\mathcal{F} = \emptyset$. 
Thus, a common perturbation algorithm can be used for both $\mathcal{D}_{\mathcal{F}}$ and $\mathcal{D}$. 



\textbf{Basic block perturbation algorithm}. 
\tool{}'s core basic block perturbation algorithm $\Gamma$ takes a set of features $\mathcal{F}\subseteq\Hat{\mathcal{P}}$ of basic block $\basicblock$ as input and randomly perturbs $\basicblock$ to obtain $\basicblock'\sim\mathcal{D}_{\mathcal{F}}$ such that $\basicblock'$ retains the features in $\mathcal{F}$ and has some of the other features in $\Hat{\mathcal{P}}$ perturbed to values valid according to the underlying ISA. Opcodes can be perturbed only to those that can accept the original set of operands, according to the ISA. Operands can be perturbed to only those having the same type and size. We elaborate on the perturbation algorithm further in this section. Figure~\ref{fig:overview}(v) shows an example perturbation of $\basicblock$ created by $\Gamma$ when preserving the features in the retained features set [Figure~\ref{fig:overview}(iv)]. While we ideally want $\mathcal{D}_{\mathcal{F}}$ to be a uniform distribution, its underlying sample space of perturbed basic blocks which preserve features in $\mathcal{F}$ is large (check Appendix~\ref{app:samplespacesize} to get an idea of the magnitude of these sample spaces) and complex without a closed-form characterization and is also defined differently for individual $\mathcal{F}$. 
This makes designing an algorithm to generate uniform samples for each $\mathcal{F}$ hard. 
Hence we relax the requirement of sampling from a uniform distribution to the ability of $\Gamma$ to produce diverse perturbed basic blocks so that the probability of obtaining a given basic block is small. Algorithm~\ref{alg:perturb_bb} in Appendix~\ref{app:pertalgo} presents the pseudocode of $\Gamma$ to 
perturb a given basic block. 

$\Gamma$ perturbs the multigraph corresponding to the basic block, $\graphbb$ to obtain $\graphbb'$, which uniquely corresponds to the perturbed basic block, such that the features in $\mathcal{F}$ are preserved.
To obtain $\graphbb'$, 
%
%
$\Gamma$ attempts to perturb every feature that is allowed to be perturbed, independent of the others. This is because any dependence will restrict the possible choices for perturbed blocks and hence disproportionately increase the probabilities of some possible perturbations. To create independence between the basic block features, $\Gamma$ perturbs vertices of the basic block graph $\graphbb$ independent of each other. $\Gamma$ also perturbs the data dependency edges that do not have any vertex in common, independent of each other. However, when two data dependency edges have at least one vertex in common if they are caused by a common operand in the instruction corresponding to the common vertex, then all perturbations to edges can not be made completely independent, otherwise they are perturbed independently. 
For independence between vertex and data dependency edge perturbations, $\Gamma$ perturbs only the opcode of the corresponding instruction of the vertex to denote vertex perturbation and only the operands of the instructions connected by edge to denote edge perturbations. 
%
$\Gamma$ preserves the opcodes of the instructions corresponding to every data dependency in $\mathcal{F}$ but can perturb other data dependencies between the instructions by operand changes. 

\emph{Perturbing vertices of $\graphbb$}.
$\Gamma$ perturbs vertices by either deleting or replacing them with other valid vertices. Deletion is permissible when the number of instructions is not required to be preserved. When a vertex is deleted, all incoming and outgoing edges of the vertex are removed from $\graphbb$. 
To replace a vertex, the corresponding instruction's opcode is replaced with another opcode in the ISA that can produce a valid assembly basic block instruction (an instruction that does not contain certain opcodes such as \opcolor{call} or \opcolor{jmp}) with the operands of the original instruction. 
%
Overall, $\Gamma$ independently perturbs or retains every vertex with equal probability, where a vertex is perturbed by either deleting or replacing it, again with equal probability.


\emph{Perturbing edges of $\mathcal{G}$}.
$\Gamma$ perturbs data dependency edges by deleting the corresponding dependency. 
The dependency is deleted by perturbing some operands corresponding to the dependency to other operands of the same type and size. The type of an operand could be memory, register, or immediate/constant, while its size could be any power of $2$ between $8-512$ bits. Hence, we change the operand registers/memory addresses to other registers/memory addresses to break the data dependencies. 
%
Overall, $\Gamma$ either perturbs or retains a data dependency by similar probabilities. The exact probabilities of perturbation and retention will be basic block specific and are discussed in Appendix~\ref{app:pertmodeldetails}.  
\textbf{Computing explanations}. With the basic block perturbation algorithm, $Prec(\mathcal{F})$ is estimated using KL-divergence-based confidence intervals~\citep{kllucb} and $Cov(\mathcal{F})$  is estimated by its empirical value, for a given set of features $\mathcal{F}\subseteq\Hat{\mathcal{P}}$. 
Similar to the Anchors' construction~\cite{anchors}, \tool{} iteratively builds its explanation feature set using a beam search wherein the maximum (estimated) precision feature sets at each level are iteratively expanded to larger feature sets till the precision threshold of $(1-\delta)$ is exceeded. The maximum coverage feature set with precision $>(1-\delta)$ is  \tool{}'s explanation for $\costmodel(\basicblock)$.

\section{Evaluation}\label{sec:evaluation}
We evaluate \tool{} to answer two main questions:

\emph{Correctness}. Do \tool{}'s explanations accurately reflect the given cost model's behavior?
    
    \emph{Utility}. Can \tool{}'s explanations be used to understand the behavior of cost models? 

\textbf{Experimental setup}. All our experiments were conducted on a 12th Gen 20-core Intel i9 processor (cache size: 24MB, RAM: 32GB, clock speed: 2.5GHz, with AVX support). We set the precision threshold $(1-\delta)$ in \eqref{eq:precision} as $0.7$. We have set the probabilities of retention and perturbation of every feature in a basic block as $0.5$. For instruction-type features where there are two possible perturbations, deletion and replacement, we assign probabilities to the perturbation operations based on an extensive hyperparameter study (Appendix~\ref{app:hyperparam}). We have used the default hyperparameters in the Anchor algorithm~\citep{anchors} for the beam-search-based iterative explanation construction method. We study the sensitivity of \tool{} to its hyperparameters in Appendix~\ref{app:hyperparam}. 
We have developed and tested \tool{} for the x86 microarchitecture. We use basic blocks from the popular BHive dataset~\citep{bhive}. We randomly pick $200$ basic blocks with number of instructions between 4 and 10 from BHive, to make our \emph{explanation test set} for testing \tool{}'s explanation. 
We run each experiment for 5 different seeds and report the average results, with their standard deviations. 

\textbf{Computing the accuracy of \tool{}'s explanations}. 
To evaluate the correctness of \tool{}'s explanations, we have developed a crude, but non-trivial,  interpretable, analytical cost model, $\crudecostmodel$. The advantage of such a model is that it gives us reliable \emph{ground truth of explanations} with which we can compare \tool{}'s explanations and compute their \emph{accuracy}. We are not aware of any actual intricate analytical cost models that have a closed-form representation that could give us ground truth explanations to objectively compute \tool{}'s accuracy, which is why we had to design $\crudecostmodel$ for \tool{}'s evaluation. 
We define $cost_{inst}(\instruction)$, $cost_{dep}(\dependency_{ij})$, and $cost_\numinsts(n)$ as the costs of the instruction $\instruction$, data dependency $\dependency_{ij}$ between instructions $i$ and $j$, and number of instructions $\numinsts = n$ respectively in a given basic block.  
\eqref{eq:crude} presents the functional form of $\crudecostmodel$. $\crudecostmodel$ identifies the features of the basic block $\basicblock$ which have the maximum cost and hence are bottlenecks for its execution, and predicts their cost as $\basicblock$'s cost. 
Our rationale behind $\crudecostmodel$ is derived from a throughput prediction baseline analytical model in \cite{uica} whose throughput prediction is the maximum of the individual costs for the number of instructions, the number of memory reads, and the number of memory writes in the input basic blocks. In $\crudecostmodel$ we have instead picked up basic block features such as its instructions and data dependencies to make the cost predictions more specific to a given block. Thus, $\crudecostmodel$ serves as a realistic, interpretable cost model, to measure the accuracy of \tool{}’s explanations. Custom $\crudecostmodel$ models can be developed for each microarchitecture when the individual cost functions vary with the microarchitecture. The exact, microarchitecture-dependent forms of the 3 cost functions used in our experiments are given in Appendix~\ref{app:crude}. 
\begin{equation}
    \label{eq:crude}
    \begin{split}
        \crudecostmodel(\basicblock) = max\{& cost_\numinsts(n), \mathop{max}_i\{cost_{inst}(\instruction_i)\}, \\ & \mathop{max}_{\dependency_{ij}}\{cost_{dep}(\dependency_{ij})\}\}
    \end{split}
\end{equation}
The ground truth explanation for $\crudecostmodel(\basicblock)$ is given by $GT(\basicblock)$~\eqref{eq:gtexpcrude}, where $type(f)$ is the type of the feature $f$ which would be one of $inst$, $dep$, and $\numinsts$. $GT(\basicblock)$ essentially is the set of basic block features that have the maximum cost among the costs for all the features. 
\begin{equation}
    \label{eq:gtexpcrude}
    GT(\basicblock) = \{f \mid f\in\Hat{\mathcal{P}}, cost_{\langle type(f)\rangle}(f) = \crudecostmodel(\basicblock)\}
\end{equation}
Note that $GT(\basicblock)$ may not be a singleton set, as there can be multiple features that are equally important and lead to the same $\crudecostmodel(\basicblock)$. We consider an explanation for $\crudecostmodel$’s prediction for $\basicblock$ to be accurate if it identifies at least one feature from $GT(\basicblock)$ and nothing outside $GT(\basicblock)$. 
%
%
We are not aware of any other competent cost model explanation methods to compare \tool{}'s accuracy against, hence we design two natural baseline explanation algorithms: \emph{random} and \emph{fixed}. The random explanation baseline includes features $f$ of $\basicblock$ based on the probability of occurrence of a feature of $type(f)$ in the set of all ground truth explanations of all basic blocks in the explanation test set. The fixed explanation baseline identifies the most frequent feature type in the set of ground truth explanations for all blocks in the explanation test set and assigns the first feature of that type in the block to be the fixed explanation.

\begin{table}
    \centering
    \caption{Accuracy of \tool{}'s explanations.}
    \begin{tabular}{@{}lrr@{}}
        \toprule
          Explanation & Acc.(\%) over $\crudecostmodel_{HSW}$  & Acc.(\%) over 
 $\crudecostmodel_{SKL}$\\
         \midrule
         Random & $26.56 \pm 20.30$ & $26.60 \pm 20.34$\\
         Fixed & $72.33$ & $74.0$\\
         \tool{} & $\mathbf{96.90 \pm 0.92}$ & $\mathbf{98.00 \pm 0.80}$\\
         \bottomrule
    \end{tabular}
    
    \label{tab:simplebaselineaccuracy}
    
\end{table} 

\begin{table}
    \centering
    \caption{Average Precision and Coverage for \tool{}'s explanations for Ithemal (I) and uiCA (U) for Haswell and Skylake.} 
    \begin{tabular}{@{}lrrr@{}}
        \toprule
          Model  & Av. Precision  & Av. Coverage\\
         \midrule
         I (HSW)  & $0.79\pm0.005$ & $0.19\pm0.007$\\
         I (SKL)  & $0.81\pm 0.004$  & $0.19\pm 0.014$\\
         U (HSW) & $0.78\pm0.006$ & $0.18\pm0.012$\\
         U (SKL) & $0.79\pm 0.006$ & $0.18\pm 0.012$ \\
         \bottomrule
    \end{tabular}
    
    \label{tab:avgpreccov}
\end{table}
\subsection{Accuracy-based evaluation of \tool{}}
Table~\ref{tab:simplebaselineaccuracy} presents the explanation accuracy achieved by \tool{} and the explanation baselines over $\crudecostmodel$ for the Haswell (HSW) and Skylake (SKL) microarchitectures. The accuracy values indicate a significant improvement in the correctness of explanations given by \tool{} over the baselines and testify the correctness of \tool{}'s explanations. 
Note that as the fixed explanation baseline does not have any randomness, it does not have any uncertainty. 

The high accuracy of \tool{}'s explanations over $\crudecostmodel$, which makes its cost predictions using the same set of features as \tool{}, indicates that \tool{} can faithfully identify the set of features that lead to the prediction when they are within the set of features that it uses to compose explanations. Note that this high accuracy has been achieved with just query access to the cost model. However, for actual cost models, it may not be the case that \tool{}'s explanation features are used directly for cost prediction. Generally, some complex functions of these basic features will be used to obtain the cost. Hence, we next estimate the precision of \tool{}'s explanations for actual cost models. 
\begin{figure*}[htb]
    \centering
    \includegraphics[width=\textwidth]{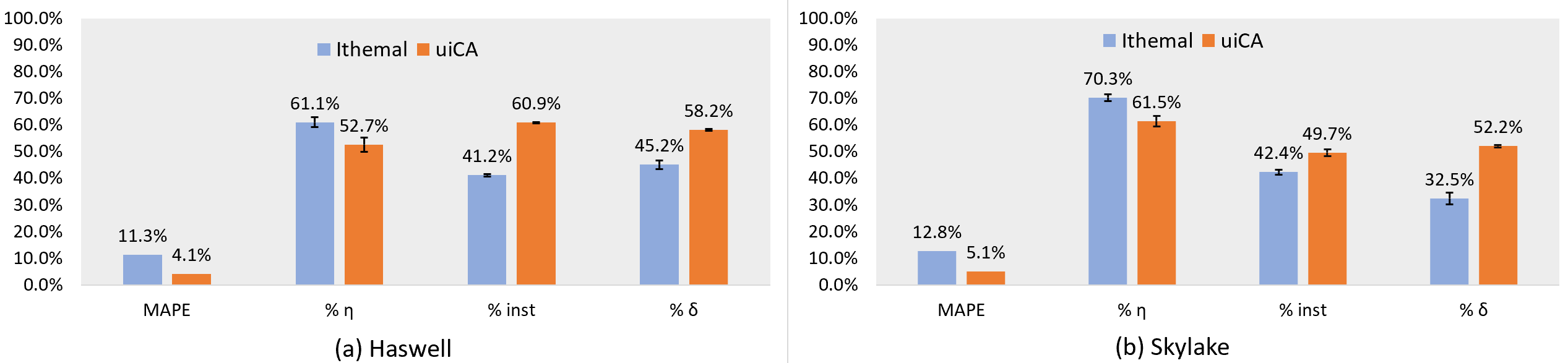}
    \caption{Variation of Mean Absolute Percentage Error (MAPE) in Ithemal and uiCA alongside variation in the \%  of explanations consisting: number of instructions $\numinsts$, specific instructions $\instruction$ and data dependencies $\dependency$. Figure (a): Haswell, (b): Skylake}
    \label{fig:utilityoverall}
\end{figure*}
\subsection{Precision and coverage evaluation} 
Next, we study the average precision and coverage of \tool{}'s explanations for state-of-the-art throughput-predicting cost models:  neural model Ithemal~\citep{ithemal}, and simulation-based model uiCA~\citep{uica} over the basic blocks in the explanation test set. We selected Ithemal and uiCA as representative cost models due to their high prediction accuracy and popularity among our stakeholders. \tool{} is applicable to other models as well, as it requires just query access to them. The average precision and average coverage are metrics to indicate \tool{}’s potential for generating faithful and generalizable explanations respectively of a target cost model for individual basic blocks in our explanation test set. 
As these cost models are not analytical, they do not have ground-truth explanations, and hence we use average precision and average coverage as proxies to evaluate the explanations, similar to~\citep{anchors}. 
The average time taken to explain a block for each model is roughly a minute. Table~\ref{tab:avgpreccov} presents our findings for Ithemal and uiCA developed for Haswell (HSW) and Skylake (SKL) microarchitectures.

We observe that the explanations for all the cost models have fairly high average precision (probability of being faithful). We estimate the coverage (generalizability) of each explanation over $10k$ perturbed basic blocks. Thus, an average coverage of $0.19$ means that the explanation for 1 block generalizes to $1900$ sampled blocks on average, making it significant for the complex domain of basic blocks. The coverage values obtained are similar to the coverage achieved while explaining models in NLP~\citep{anchors}. These results indicate that the high accuracy of \tool{} over our custom cost model $\crudecostmodel$ transfers to state-of-the-art cost models as well and \tool{} can be deployed to obtain high-quality explanations for common cost models. Next, we study how \tool{} can become an essential model analysis tool for our stakeholders.


\begin{figure*}[htb]
    \centering
    \includegraphics[width=\textwidth]{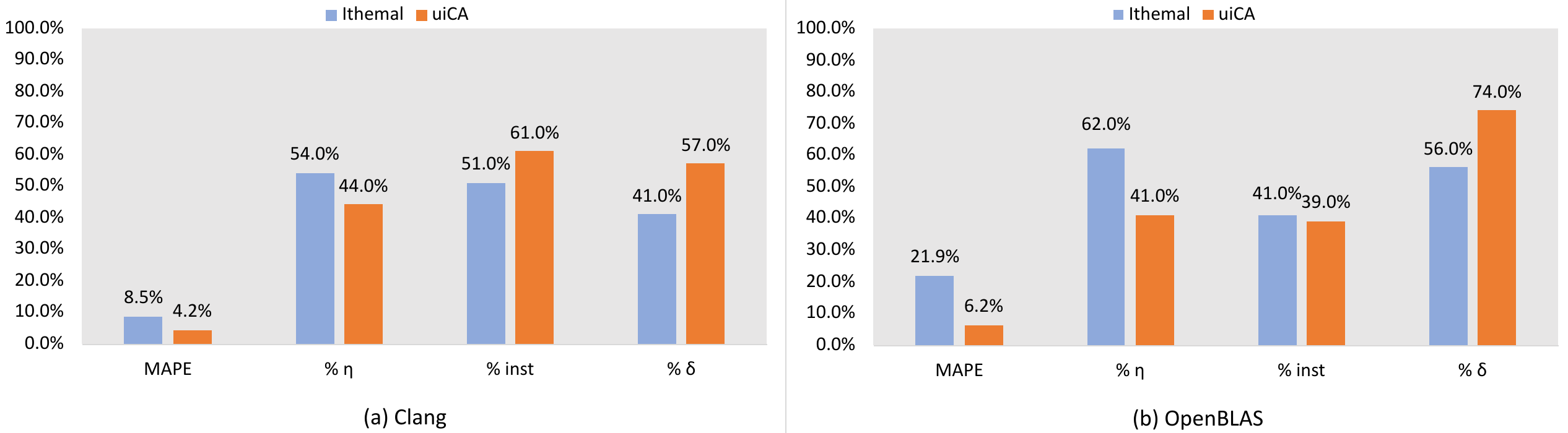}
    \caption{Variation of Mean Absolute Percentage Error (MAPE) in Ithemal and uiCA with the \% of explanations consisting: number of instructions $\numinsts$, instructions $\instruction$ and data dependencies $\dependency$. BHive sources: (a) Clang, (b) OpenBLAS}
    \label{fig:bhivesources}
\end{figure*}

\begin{figure*}[h]
    \centering
    \includegraphics[width=\textwidth]{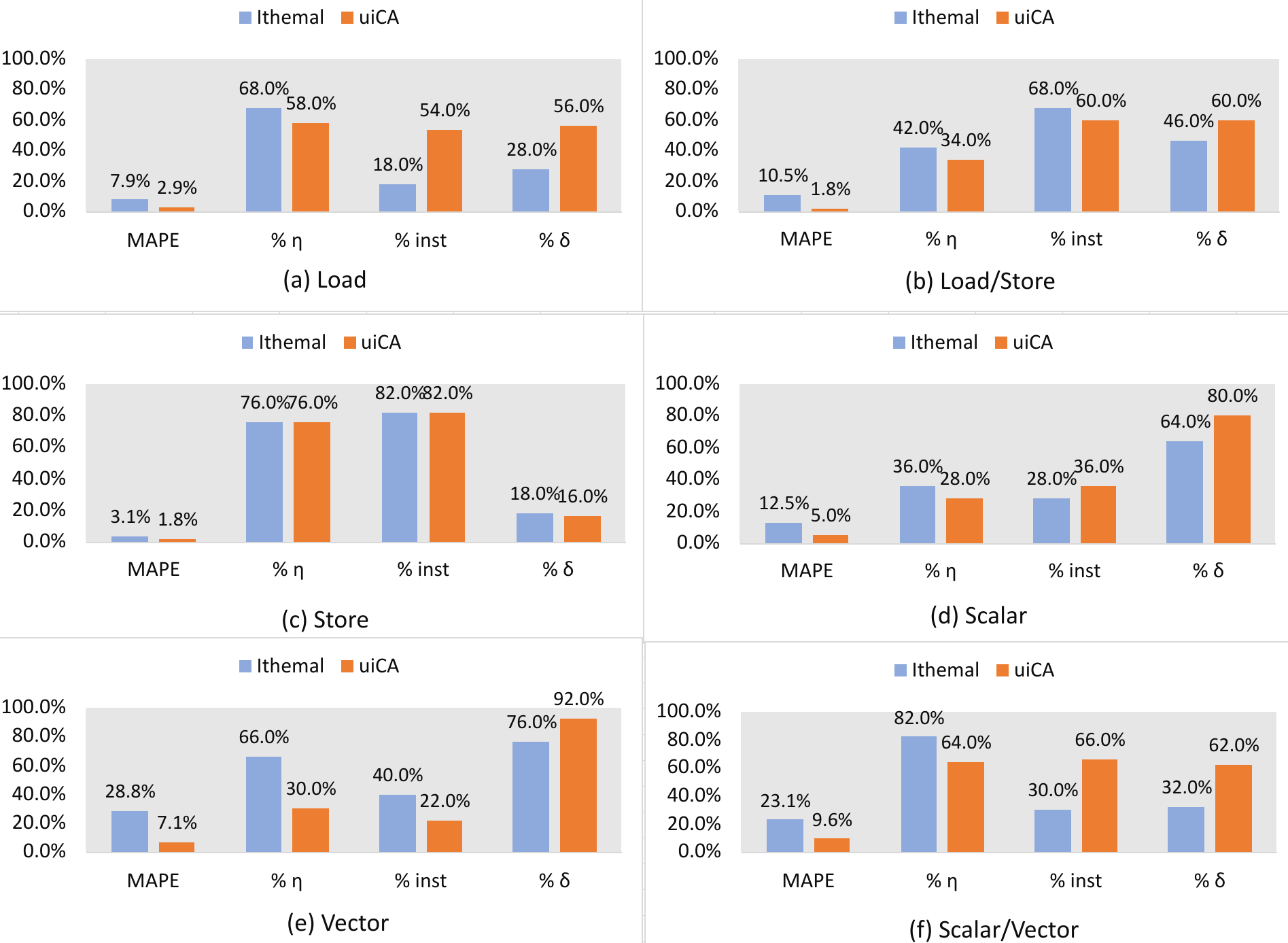}
    \caption{ Mean Absolute Percentage Error (MAPE) in Ithemal and uiCA and the \% of explanations having: number of instructions $\numinsts$, instructions $\instruction$ and data dependencies $\dependency$. BHive categories: (a) Load,  (b) Load/Store, (c) Store, (d) Scalar,  (e) Vector, (f) Scalar/Vector}
    \label{fig:bhivecats}
\end{figure*}

\subsection{Evaluating utility of \tool{}}\label{sec:cometutileval}
We show a use case of \tool{} wherein we investigate the variation in the prediction errors of Ithemal and uiCA and empirically check its correlation with the dependence of the model's output on different types of block features. We hypothesize that as the error of the cost model decreases, its dependence on the finer-grained features of the block will increase. Out of the $3$ types of features over which \tool{} composes its explanations, we identify the block's instructions and data dependencies as more specific, finer-grained features when compared to the feature corresponding to the number of instructions in the block. We use \tool{}'s explanations to identify the block features on which the model's prediction depends. Note that our hypothesis is not obvious even for arbitrary traditional cost models (let alone the black-box neural cost models), as is evident from the performance of a baseline cost model in \citet{uica}, Table 1, that uses coarse-grained features of a basic block and achieves higher accuracy than LLVM-MCA~\citep{llvm-mca} which uses finer-grained features. 
Figure~\ref{fig:utilityoverall} shows the results of our investigation. It shows the variation of mean absolute percentage error of Ithemal and uiCA. Alongside the error, it shows the percentage of \tool{}'s explanations over the entire explanation test set that contain features corresponding to the number of instructions $\numinsts$, instructions $\instruction$, and data dependencies $\dependency$ in the explained basic block. 
 

The trends in Figure~\ref{fig:utilityoverall} for both Haswell and Skylake confirm our hypothesis.
 Interpret this insight as follows: as the cost model becomes more accurate, it focuses more on the finer-grained features of the basic block, as indicated by \tool{}'s explanations. We discuss similar insights obtained for blocks derived from different partitions of the BHive dataset, as described in Appendix~\ref{sec:bhive} next. 
%
We omit the error bars for clarity, as the standard deviations in our results are generally low [Figure~\ref{fig:utilityoverall}]. 

\textbf{BHive partitions by source}. 
We study the explanations for blocks in BHive derived from the Clang and OpenBLAS sources. We select $100$ unique blocks from each source to separately analyze our hypothesis. Figure~\ref{fig:bhivesources} presents our findings and confirms our hypothesis for both partitions.

\textbf{BHive partitions by category}.
We conduct a similar study on 50 unique basic blocks corresponding to each category in the BHive dataset. Figure~\ref{fig:bhivecats} presents our findings and confirms our hypothesis for all categories. 
Interestingly, for the \emph{Store} category, as the error in throughput predictions of both cost models is similar, we observe similar prominence of all types of features in \tool{}'s explanations for both cost models. 
This observation further supports our hypothesis.


\subsection{Case studies}\label{app:casestudies}
Next, we show another use case of \tool{}'s explanations \textemdash to conduct analyses of cost prediction of individual basic blocks. Similar analyses can be useful to understand the cost model's behavior in corner cases. We discuss \tool{}'s explanations for the predictions of  Ithemal and uiCA for Haswell on randomly picked blocks from the BHive dataset.


\textbf{Case study 1}. The block in Listing~\ref{code:cs1} is predicted to have a throughput of 2 cycles by both cost models which matches the throughput on actual hardware reported in the BHive dataset. 
Instructions 2 and 3 write to the memory and are thus the highest throughput instructions~\cite{uops.info,agnerfog}. Hence intuitively, for correct prediction, these instructions are important. \tool{}'s explanations for both cost models match this intuition, thus suggesting that both cost models actually consider the intuitive set of features to correctly predict throughput for this block.
\begin{code}
\begin{center}
\begin{minipage}{0.4\textwidth}
\begin{lstlisting}[language={[x86masm]Assembler}]
lea rdx, [rax + 1]
mov qword ptr [rdi + 24], rdx
mov byte ptr [rax], 80
mov rsi, qword ptr [r14 + 32]
mov rdi, rbp
\end{lstlisting}
\end{minipage}

\begin{tabular}{@{}llr@{}}
     & \textbf{Prediction} & \textbf{Explanation}\\
     \textbf{Ithemal} & 2 cycles & $\left\{\instruction_2, \instruction_3\right\}$\\
     \textbf{uiCA} & 2 cycles & $\left\{\instruction_2, \instruction_3\right\}$\\
\end{tabular}
\captionof{listin}{Case Study 1}
\label{code:cs1}
\end{center}
\end{code}

\textbf{Case study 2}.
The block in Listing~\ref{code:cs2} has a division instruction and many data dependencies such as a RAW data dependency between instructions 3 and 6 due to register \opcolor{rax} and a WAR dependency between instructions 1 and 2 due to register \opcolor{edx}. A div instruction is a very expensive instruction in general~\cite{uops.info,agnerfog}. The actual throughput of the basic block is 39 cycles. Thus, both cost models have made incorrect predictions, but the prediction of Ithemal is more erroneous as compared to uiCA. \tool{}'s explanation for Ithemal consists of just the feature corresponding to the number of instructions in the basic block, while that for uiCA consists of the \opcolor{div} instruction and a data dependency. These explanations suggest that Ithemal does not sufficiently prioritize costly instructions such as \opcolor{div} and data dependencies, unlike the actual microarchitecture that Ithemal is trained to mimic, thus indicating potential sources of its throughput-prediction error. 

\begin{code}
\begin{center}
\begin{minipage}{0.4\textwidth}
\begin{lstlisting}[language={[x86masm]Assembler}]
mov ecx, edx
xor edx, edx
lea rax, [rcx + rax - 1]
div rcx
mov rdx, rcx
imul rax, rcx
\end{lstlisting}
\end{minipage}

\begin{tabular}{@{}llr@{}}
     & \textbf{Prediction} & \textbf{Explanations}\\
     \textbf{Ithemal} & 23 cycles & $\left\{\numinsts (num\_insts)\right\}$\\
     \textbf{uiCA} & 36 cycles & $\left\{\dependency_{RAW,3,6},\instruction_4\right\}$\\
\end{tabular}
\captionof{listin}{Case Study 2}
\label{code:cs2}
\end{center}
\end{code}

    

\section{Discussion and future work}
We demonstrated how \tool{}'s explanations can be used to gain both high-level [Section~\ref{sec:cometutileval}] and case-specific [Section~\ref{app:casestudies}] insights about cost models and compare their behaviors against other cost models. These insights can be useful for repairing high-error neural models with domain-specific insights and developing more generalizable models in the future. As indicated by these insights, neural architectures that explicitly utilize the finer-grained features of blocks can achieve better cost prediction performance. 
\tool{}'s explanations can be used to select a model from a collection of similar performing neural models. 
\tool{} can be extended to run on GPUs to make it amenable to integration with cost model training and inference procedures, in the future. \tool{}'s feedback can be leveraged to update the model parameters during training to have the predictions rely on finer-grained features. \tool{} can be augmented to existing cost models to guide compiler optimizations with information on what parts of the basic block need to be optimized for better performance.

While explanation features employed by \tool{} currently capture the commonly used properties of a block, it will produce approximations for the most important factors behind a model's predictions when they cannot be captured by the current features. We will investigate expanding the explanation features in future work. 
Finally, \tool{} can be extended to other open-source ISAs, including those for GPUs and TPUs, by mapping the current perturbation algorithm to the new ISA. We need to define the opcodes (operands) that could replace each opcode (operand) to generate a valid perturbation. While the high-level formalism can be carried over, instance-specific challenges can arise.

\section{Conclusion}
We presented \tool{}, the first approach for efficiently generating faithful, generalizable, and interpretable explanations for neural cost models. Our results show that \tool{} can generate accurate and useful explanations that indicate potential sources of errors. 
We believe that \tool{}'s explanations can be used to improve trust in the workings of neural cost models and accelerate their real-world adoption.  
\section{Acknowledgement}
We thank the anonymous reviewers for their insightful comments. This work was supported in part by
NSF Grants No. CCF-2238079, CCF-2316233, CNS-2148583, and Google Research Scholar award.

\bibliography{refs}
\bibliographystyle{plainnat}

\appendix
\newpage
\newpage
\section{Monotonicity of perturbation function}\label{app:proionofs}

\begin{theorem}
    $\pertmodel$ is a monotonically decreasing function.
\end{theorem}

\begin{proof}
    Let $F_1,F_2\in\wp(\mathcal{P})$ such that $F_1\subseteq F_2$.  

    \begin{equation*}
        \begin{split}
            \pertmodel(F_1) =& \{\basicblock'\mid\basicblock'\in B, \\ & F_1\subseteq\mathcal{P}_{\basicblock'}, \mathcal{P}_{\basicblock'}\setminus F_1\text{ are obtained from }\mathcal{P}\setminus F_1\}\\
            =& \{\basicblock'\mid\basicblock'\in B, F_2\subseteq\mathcal{P}_{\basicblock'},\\ &\mathcal{P}_{\basicblock'}\setminus F_2\text{ are obtained from }\mathcal{P}\setminus F_2\}\\
            & \cup\{\basicblock'\mid\basicblock'\in B, F_1\subseteq\mathcal{P}_{\basicblock'},  F_2\not\subseteq\mathcal{P}_{\basicblock'}, \\ & \mathcal{P}_{\basicblock'}\setminus F_1\text{ are obtained from }\mathcal{P}\setminus F_1\}\\ 
            =& \pertmodel(F_2)\cup\{\basicblock'\mid\basicblock'\in B, F_1\subseteq\mathcal{P}_{\basicblock'},  F_2\not\subseteq\mathcal{P}_{\basicblock'}, \\ & \mathcal{P}_{\basicblock'}\setminus F_1\text{ are obtained from }\mathcal{P}\setminus F_1\}
        \end{split}
    \end{equation*}

    Hence, $\pertmodel(F_2)\subseteq\pertmodel(F_1)$
\end{proof}

Note that in the above proof, features in feature sets such as $\mathcal{P}_{\basicblock'}\setminus F_1$ are obtained by either retaining or perturbing the features in $\mathcal{P}\setminus F_1$.

A similar proof can be used to prove the monotonicity of $\Hat{\pertmodel}$ as well. 

\section{Types of data dependencies in basic blocks}\label{app:deps}
While each instruction is processed sequentially by the different components of the CPU, an instruction $\instruction_j$ can get stalled due to the requirement for a previous instruction $\instruction_i$ to get completed, creating a \emph{data dependency hazard}~\citep{comporgdes}. A Read-After-Write (RAW) data-dependency hazard arises when $\instruction_j$ reads the value in an operand that is written by $\instruction_i$. $\instruction_j$ can not get executed until $\instruction_i$ ends to ensure correct execution. 
A Write-After-Read (WAR) hazard occurs when $\instruction_j$ writes to an operand that is read by $\instruction_i$.
A Write-After-Write (WAW) hazard arises when $\instruction_j$ writes to an operand that is written to by $\instruction_i$. 
There can be multiple data dependency hazards, possibly of different kinds, between a given pair of instructions. 

\section{Basic Block Perturbation Algorithm}\label{app:pertalgo}
Algorithm~\ref{alg:perturb_bb} presents our stochastic perturbation algorithm $\Gamma$ to conditionally perturb a given basic block $\basicblock$ to $\basicblock'$. The perturbation algorithm creates the graph $\graphbb'$ of $\basicblock'$ while preserving a set of instructions/their corresponding vertices $\overline{\mathcal{V}}$, a set of data dependencies/their corresponding edges $\overline{\edges}$ and possibly the number of instructions/the number of vertices, denoted by the boolean $preserve_\numinsts$ which is set to true when the number of instructions $\numinsts$ is to be kept constant. If the number of vertices is to be kept constant, then the vertex/instruction deletion operation is forbidden [lines~\ref{alg:numinststart}-\ref{alg:numinstend}]. The vertices at the ends of the edges in $\overline{\edges}$ are preserved as well [line~\ref{alg:addvertices4deps}] by adding them to $\overline{\mathcal{V}}$. Then each vertex of $\graphbb$ is perturbed with a probability of $(1-p_{I,ret})$ if it is not required to be retained [lines~\ref{alg:nodeperturbstart}-\ref{alg:nodeperturbend}]. If the deletion perturbation operation is in vertexPertOps, then a vertex is deleted or replaced with probabilities of $p_{del}$ and $(1-p_{del})$ respectively. Otherwise, it is replaced with a valid vertex. The replacement of a vertex/corresponding instruction involves changing its opcode to another opcode that can take the original operands and still constitute valid x86 syntax according to the x86 Instruction Set Architecture. Similarly, each data-dependency edge is perturbed with a probability of $(1-p_{D,ret})$ if it is not required to be retained [lines~\ref{alg:edgeperturbstart}-\ref{alg:edgeperturbend}], to form $\graphbb'$ [line~\ref{alg:makepertbb}]. The only perturbation of any data dependency is its deletion, which is conducted by the perturbation of the operands involved in the data dependency. 

\begin{algorithm}
   \caption{Basic Block Perturbation Algorithm}
   \label{alg:perturb_bb}
\begin{algorithmic}[1]
   \STATE {\bfseries Input:} basic block graph $\graphbb$, vertices to preserve $\overline{\mathcal{V}}$, data-dependency edges to preserve $\overline{\edges}$, $preserve_\numinsts$, $p_{I,ret}$, $p_{D,ret}$, $p_{del}$
   \STATE {\bfseries Output:} perturbed basic block graph, $\graphbb'$
   \STATE vertexPertOps = \{replacement, deletion\}
   \IF{$preserve_\numinsts$}\label{alg:numinststart}
        \STATE vertexPertOps.remove(\{deletion\})
   \ENDIF\label{alg:numinstend}
   \STATE $\overline{\mathcal{V}} \gets addVerticesForPreservedDeps(\overline{\mathcal{V}}, \overline{\edges})$\label{alg:addvertices4deps}
   \FOR{$\node \in GetVertices(\basicblock)$}\label{alg:nodeperturbstart}
        \IF{$v \not\in \overline{\nodes}$ \AND $rand([0,1])> p_{I,ret}$}
            \STATE $\node \gets PerturbVertex(\graphbb, \node, vertexPertOps, p_{del})$
        \ENDIF
   \ENDFOR\label{alg:nodeperturbend}
   \FOR{$\edge \in GetDepEdges(\basicblock)$}\label{alg:edgeperturbstart}
       \IF{$\edge \not\in \overline{\edges}$ \AND $rand([0,1])> p_{D,ret}$}
            \STATE $\edge \gets PerturbEdge(\graphbb, \edge)$
        \ENDIF
   \ENDFOR\label{alg:edgeperturbend}
   \STATE $\graphbb' \gets \graphbb$\label{alg:makepertbb}
\end{algorithmic}
\end{algorithm}

\section{Case specificity of perturbation probabilities} \label{app:pertmodeldetails}
\tool{}'s perturbation algorithm $\Gamma$ consists of primarily 3 probability terms: $p_{I,ret}$, $p_{D,ret}$, and $p_{del}$ as described in Appendix~\ref{app:pertalgo}. $p_{I,ret}$ and $p_{D,ret}$ are the probabilities of retention of a given instruction and a given data dependency respectively, in the perturbed basic block. $p_{del}$ is the probability of deletion of an instruction when the deletion perturbation operation is allowed for instructions. 
The deletion perturbation operation will not be allowed for instructions when the number of instructions is to be kept constant.

$\Gamma$ perturbs a basic block $\basicblock$ by essentially perturbing every instruction while preserving certain tokens of the instruction from getting perturbed. These preserved tokens correspond to the features that are required to be preserved by $\Gamma$ and also the features that $\Gamma$ voluntarily does not attempt to perturb. $\Gamma$ has voluntary retention of randomly selected basic block features to output perturbed basic blocks $\basicblock'$ that are very similar to the original basic block $\basicblock$. $\Gamma$ attempts to perturb the other tokens of $\basicblock$ to obtain $\basicblock'$. 

$\Gamma$ can delete an instruction in case none of its tokens are required to be preserved. Otherwise, $\Gamma$ replaces a token with another token that can form a basic block with valid x86 syntax alongside the other tokens. Thus, every token has a set of potential replacements. Perturbations to opcode tokens are counted as changes to the instruction features, while perturbations to the operand tokens are considered as changes to any data dependency features.
As the perturbation space consists of only valid basic blocks, the overall probabilities of the primitive perturbation operations (instruction deletion, instruction replacement, and data dependency deletion) vary with the target basic block. 

Following is an example of this variation. Several tokens of x86 assembly have no possible replacements resulting in no probability of replacement, such as the opcode \opcolor{lea}. This is a special opcode that loads the effective memory address of its source operand into the destination register. There is no other x86 opcode that shows similar behavior. Hence, the \opcolor{lea} can not be replaced with any other opcode. Such failed attempts at opcode replacement lead to the retention of the instruction, thus leading to an increase in the probability of retention of specific features of the basic block. This change in probabilities is specific to the basic blocks having the \opcolor{lea} opcode in its instructions. 

Another example of basic-block-specific probability settings occurs due to data dependencies. The data dependencies in a basic block can be varied with changes in just the opcodes of the corresponding instructions. Thus, while we keep the perturbation probability of a data dependency $(1-p_{D,ret})$ to be $0.5$ in the general case, it can vary with the basic block. A basic block having all the potential replacements for the opcodes involved in a data dependency with similar behavior as the original opcodes will have $0.5$ probability of perturbation of the data dependency, while the opcodes for which we have potential replacements show variable behaviors, the data dependency perturbation probability can be more than $0.5$. (Opcodes \opcolor{add} and \opcolor{sub} have similar behavior as they read the value in the source operand and read-write the value in the destination operand. They have different behavior from \opcolor{mov} that reads the source operand value and writes to the destination operand. All 3 opcodes could be potential replacements for each other in instructions having certain pairs of operands.)



\section{Ablation and Sensitivity Studies}\label{app:hyperparam}
In this section, we study the variations in our results, with \tool{}'s hyperparameters and design choices. 
We use our explanation accuracy-based evaluation scheme based on our crude but interpretable cost model that is presented in Section~\ref{sec:evaluation}, to study the effects of the different hyperparameters and design choices. For this study, we have used the crude cost model for the Haswell microarchitecture. We have randomly selected 100 basic blocks from the BHive dataset~\cite{bhive} for which we generate \tool{}'s explanations with different settings. We have dropped the error bars for clarity of the results, as we note from Table~\ref{tab:simplebaselineaccuracy} that the standard deviations in our results are generally low.

\subsection{Precision threshold}
In this section, we study the variation in the explanations' accuracy with the precision threshold set in \tool{}, above which we consider the explanation feature set to be approximately faithful to the cost model's predictions. We want the precision threshold to be high such that the most precise and accurate explanations are given as output. 
Figure~\ref{fig:abl_precthresh} presents the variation in the accuracy of \tool{}'s explanations with various values for the precision threshold $(1-\delta)$ in \tool{}. We observe that $0.7$ is the highest precision threshold that gives the highest accuracy and hence we have set it as the precision threshold in our experiments. 

\begin{figure}
    \centering
    \includegraphics[width=0.5\textwidth]{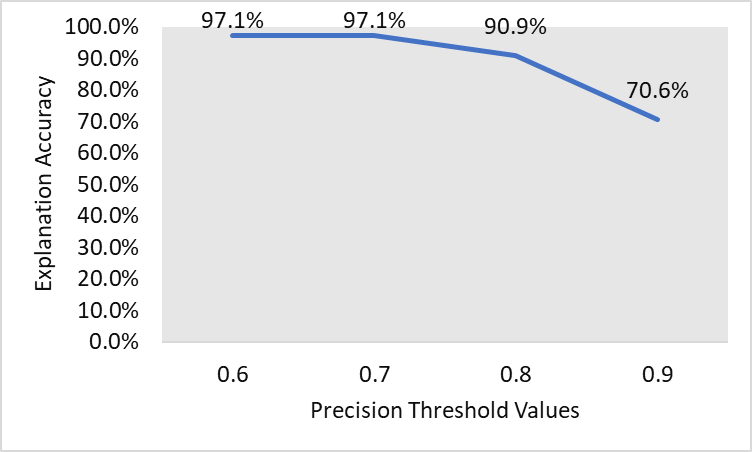}
    \caption{Variation in explanation accuracy with the precision threshold $(1-\delta)$ setting in \tool{}}
    \label{fig:abl_precthresh}
\end{figure}

\subsection{Perturbation probabilities for instructions}
$\Gamma$ attempts to perturb a given instruction $\instruction$ in a basic block $\basicblock$ only when it is not required to be preserved. $\Gamma$ retains $\instruction$ with a probability of $p_{I,ret}$ and perturbs it otherwise. 
There are 2 potential operations for perturbing $\instruction$: Deletion and Replacement (with valid x86 instruction), each probabilities $p_{del}$ and $(1-p_{del})$ respectively. We have set $p_{del} = 0.33$ based on a sensitivity study that we conducted with respect to this hyperparameter, for all of our experiments. Figure~\ref{fig:abl_instdel} presents our findings. We find that our choice of $p_{del} = 0.33$ leads to the maximum accuracy among other candidates. 

\begin{figure}
    \centering
    \includegraphics[width=0.5\textwidth]{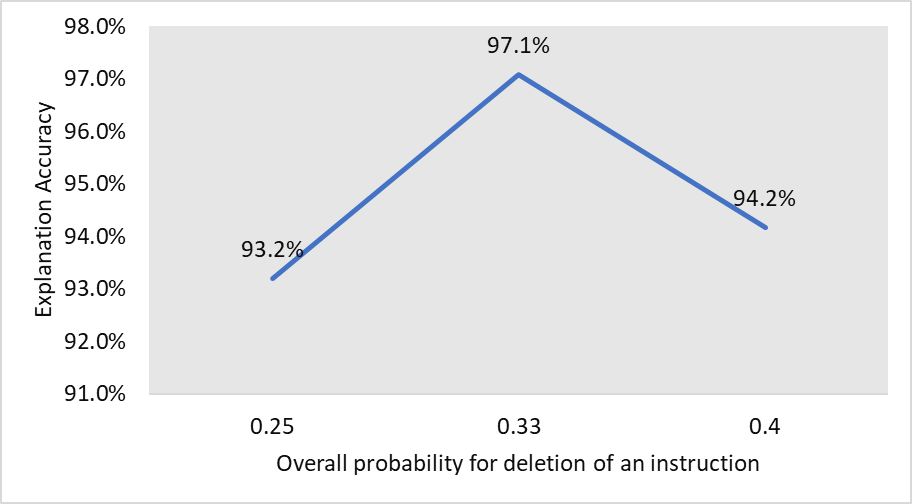}
    \caption{Variation in explanation accuracy with the probability of instruction deletion in $\Gamma$}
    \label{fig:abl_instdel}
\end{figure}

\subsection{Perturbation probabilities for data dependencies}
Similar to the case for instructions, $\Gamma$ attempts to perturb a given data dependency $\dependency$ in a basic block $\basicblock$ with probability $(1-p_{D,ret})$. As discussed in Section~\ref{app:pertmodeldetails}, the exact probabilities of the retention/deletion of data dependencies are basic-block-specific. However, we vary these probabilities by varying the probability of explicit retention of a data dependency, i.e. the probability by which a data dependency will be retained for sure. This probability is a lower bound for $p_{D,ret}$ and higher values of this lower bound imply higher values for $p_{D,ret}$ for any given basic block. Figure~\ref{fig:abl_datadeps} shows our findings. We have shown the variation in explanation precision as well, as we observe precision to have a trend different from explanation accuracy in this case. We find that a value of $0.1$ for this probability parameter leads to optimum values for both explanation accuracy and precision. 
Thus, we have selected the explicit data dependency retention probability to be $0.1$ in \tool{}. 

\begin{figure}
    \centering
    \includegraphics[width=0.5\textwidth]{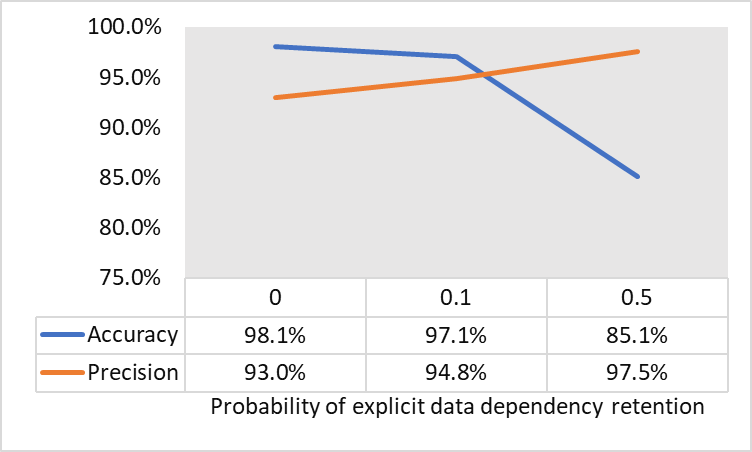}
    \caption{Variation in explanation accuracy and precision with the probability of explicit data dependency retention}
    \label{fig:abl_datadeps}
\end{figure}

\subsection{Replacement of instructions}
$\Gamma$ considers only the changes to an instruction's opcode as changes to the feature corresponding to the instruction. 
However, another possibility could be to consider operand changes (such that their types and sizes are preserved) as well as changes to the instruction feature. We analyze the effects of the two instruction changing/replacement schemes in Figure~\ref{fig:abl_instrep}. We observe that the accuracy of the explanations is higher with just the opcode replacement method, justifying our choice of this instruction replacement scheme. 

\begin{figure}
    \centering
    \includegraphics[width=0.5\textwidth]{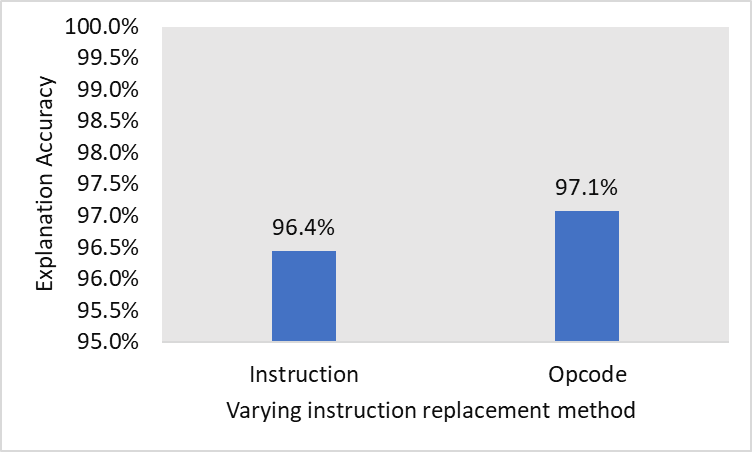}
    \caption{Variation in explanation accuracy with just opcode and whole instruction replacement schemes.}
    \label{fig:abl_instrep}
\end{figure}

An important hyperparameter that we have set according to our intuitive understanding is the $\epsilon$ error, which marks the radius of the ball of acceptable cost predictions around the prediction of cost model $\costmodel$ for basic block $\basicblock$ ($\costmodel(\basicblock)$). For our crude cost model $\crudecostmodel$, we have kept $\epsilon$ to be a quarter of one unit of its cost prediction, as the least change in its cost prediction can be a quarter unit ($\frac{\Delta n}{4} = 0.25$). For the practical cost models such as Ithemal and uiCA, we have set $\epsilon$ as $0.5$ cycles of throughput prediction, as that is the least, significant change in practically-useful throughput values.  

\section{Perturbation function output sizes}\label{app:samplespacesize}
The perturbation function, $\pertmodel_\basicblock:\wp(\mathcal{P}_\basicblock)\rightarrow\wp(\mathcal{B})$ maps a given set of basic block features $\mathcal{F}$ to the set of basic blocks $\mathcal{B}_{\mathcal{F}}$ that have $\mathcal{F}$ and where the other features are obtained from perturbations to the features in $\mathcal{P}_\basicblock\setminus\mathcal{F}$. In this section, we provide estimates of cardinalities of $\mathcal{B}_{\mathcal{F}}$ for some basic blocks $\basicblock$ and feature sets $\mathcal{F}$. With this analysis, we allude to the practical intractability of generating ideal black-box explanations for cost models.

Note that, as $\mathcal{P}_\basicblock$ is the set of all features (all basic features and all of their functions) of $\basicblock$, it can be an infinite set itself. $\Hat{\mathcal{P}}_\basicblock\subset\mathcal{P}_\basicblock$, hence for $\mathcal{F}\subseteq\Hat{\mathcal{P}}_\basicblock$, $\Hat{\pertmodel}_\basicblock(\mathcal{F})\subseteq\pertmodel_\basicblock(\mathcal{F})$. Hence, $\lvert\Hat{\pertmodel}_\basicblock(\mathcal{F})\rvert\leq\lvert\pertmodel_\basicblock(\mathcal{F})\rvert$. Thus, we provide estimates for $\lvert\pertmodel_\basicblock(\mathcal{F})\rvert$ by reporting the rough values for $\lvert\Hat{\pertmodel}_\basicblock(\mathcal{F})\rvert$. 

First, consider the basic block $\basicblock_1$ in Listing~\ref{code:pfsize1}, for $\mathcal{F}=\emptyset$. $\lvert\Hat{\pertmodel}_{\basicblock_1}(\emptyset)\rvert \approx 1.94\times 10^{38}$. As we add more elements to $\mathcal{F}$, the size of $\lvert\Hat{\pertmodel}_{\basicblock_1}(\mathcal{F})\rvert$ will reduce due to the constraints introduced to the perturbations. 

\begin{code}
\begin{center}
\begin{minipage}{0.4\textwidth}
\begin{lstlisting}[language={[x86masm]Assembler}]
vdivss xmm0, xmm0, xmm6
vmulss xmm7, xmm0, xmm0
vxorps xmm0, xmm0, xmm5
vaddss xmm7, xmm7, xmm3
vmulss xmm6, xmm6, xmm7
vdivss xmm6, xmm3, xmm6
vmulss xmm0, xmm6, xmm0
\end{lstlisting}
\end{minipage}
\captionof{listin}{Basic block $\basicblock_1$ for perturbation function size estimation}
\label{code:pfsize1}
\end{center}
\end{code}

Next, for $\mathcal{F}=\{\instruction_1\}$ i.e. with no perturbations to instruction 1 in $\basicblock_1$, $\lvert\Hat{\pertmodel}_{\basicblock_1}(\mathcal{F})\rvert \approx 6.58\times 10^{29}$. 

Similarly, consider the basic block $\basicblock_2$ in Listing~\ref{code:pfsize2}, for $\mathcal{F}=\emptyset$. $\lvert\Hat{\pertmodel}_{\basicblock_2}(\emptyset)\rvert \approx 1.63\times 10^{32}$. For $\mathcal{F}=\{\instruction_2\}$ i.e. with no perturbations to instruction 2 in $\basicblock_2$, $\lvert\Hat{\pertmodel}_{\basicblock_2}(\mathcal{F})\rvert \approx 2.77\times 10^{28}$. 

\begin{code}
\begin{center}
\begin{minipage}{0.4\textwidth}
\begin{lstlisting}[language={[x86masm]Assembler}]
shl eax, 3
imul rax, r15
xor edx, edx
add rax, 7
shr rax, 3
lea rax, [rbp + rax - 1]
div rbp
imul rax, rbp
mov rbp, qword ptr [rsp + 8]
sub rbp, rax
\end{lstlisting}
\end{minipage}
\captionof{listin}{Basic block $\basicblock_2$ for perturbation function size estimation}
\label{code:pfsize2}
\end{center}
\end{code}

Thus, we find that the perturbation function's output set can have very high cardinality, posing a challenge for generating desirable explanations. 

\section{Crude interpretable cost model details}\label{app:crude}
We define $cost_{inst}(\instruction)$ as the throughput of the instruction $\instruction$ on actual hardware. We obtain the throughputs of instructions over actual hardware from \url{https://www.uops.info/table.html}. 
We define $cost_{dep}(\dependency_{ij})$ as in~\eqref{eq:depcost}. Our intuition behind keeping the costs of WAR and WAW type of dependencies to be $0$ is that these dependencies are not true dependencies and can be generally resolved by the compiler by register renaming~\cite{comporgdes}. The RAW data dependency, on the other hand, is a true dependency. As the two instructions forming a RAW dependency will be executed sequentially on hardware, the addition of their individual costs would be a good proxy for the actual throughput cost brought in by the data dependency. 
\begin{equation}
\label{eq:depcost}
    \begin{split}
        &cost_{dep}(\dependency_{ij})  \\ 
    &=\small{\begin{cases}
        0, &\dependency_{ij}=\text{WAR/WAW}\\
        cost_{inst}(\instruction_i) + cost_{inst}(\instruction_j), &\dependency_{ij}=\text{RAW}
    \end{cases}}
    \end{split}
\end{equation}
We define the $cost_\numinsts(n) = \numinsts/4$ as the cost for having $n$ number of instructions (denoted by $\numinsts$) in a given basic block $\basicblock$. We derive the expression for the cost of number of instructions from the simple baseline model presented in \cite{uica}. 

Our choice of $\crudecostmodel$ is microarchitecture-specific as the costs of individual instructions vary across microarchitectures. We have developed $\crudecostmodel$ models for the Haswell and Skylake microarchitectures, only for the purposes of evaluating \tool{}'s explanations. 

\section{Studied dataset and cost models}\label{app:ithemaluicabhive}
\subsection{BHive dataset}\label{sec:bhive}
BHive dataset\footnote{\url{https://github.com/ithemal/bhive}}~\cite{bhive} is a benchmark suite of x86 basic blocks. It contains roughly 300,000 basic blocks annotated with their average throughput over multiple executions on actual hardware for 3 microarchitectures: Haswell, Skylake, and Ivy Bridge. We have generated explanations for basic blocks in this dataset.

The dataset can be partitioned by 2 criteria: by \emph{source} and by \emph{category} of its basic blocks. Partition by source annotates each block with the real-world code base from which it has been derived. Examples of BHive sources are Clang and OpenBLAS. 
Partition by category annotates each basic block by its type, characterized by the semantics of the instructions in the block. There are 6 types of blocks: Scalar, Vector, Scalar/Vector, Load, Store, and Load/Store. 

\subsection{Ithemal} \label{sec:ithemal}
Ithemal\footnote{\url{https://github.com/ithemal/Ithemal}}~\cite{ithemal} is an ML-based cost model, which predicts the throughput of input x86 basic blocks for a given microarchitecture. It is open-source and is currently trained for the Haswell, Skylake, and Ivy Bridge microarchitectures on the BHive dataset. A separate instance of Ithemal needs to be trained for every microarchitecture, due to the difference in the actual throughput values obtained over different hardware. Ithemal's throughput prediction is a floating point number, as it is trained on the BHive dataset. 

Ithemal consists of a hierarchical multiscale RNN structure. The first RNN layer takes embeddings of tokens of the input basic block and combines them to create embeddings for the instructions in the basic block. The second RNN layer takes the instruction embeddings as input and combines them to create an embedding for the basic block. The basic block embedding is passed through a linear regressor layer to compute the throughput prediction for the basic block. 

Ithemal exhibits roughly 9\% Mean Absolute Percentage Error for the Haswell microarchitecture on the BHive dataset. As Ithemal outputs only its throughput prediction and no insights into why the prediction was made, it can not be reliably deployed in mainstream compiler optimizations. 

\subsection{uiCA} \label{sec:uica}
uiCA\footnote{\url{https://github.com/andreas-abel/uiCA}}~\cite{uica} is an analytical simulation-based cost model for several latest microarchitectures released by Intel over the last decade. uiCA’s simulation model is hand-engineered to accurately match the model of each Intel microarchitecture and must be manually tuned to reflect new microarchitectures. It can output detailed insights into its process of computing its throughput prediction of input x86 basic blocks, such as where in the CPU's pipeline its simulator identified a bottleneck for the execution of the basic block. 





\end{document}